\documentclass[9pt,conference]{IEEEtran}
\IEEEoverridecommandlockouts
\usepackage{cite}
\usepackage{amsmath,amssymb,amsfonts}
\usepackage{algorithmic}
\usepackage{graphicx}
\usepackage{textcomp}
\usepackage{xcolor}

\usepackage[T1]{fontenc}
\usepackage[utf8]{inputenc}
\usepackage[ngerman,english]{babel}
\usepackage{amsmath, amssymb, ﻿amsfonts}   % allows usage of various math symbols and psfrag
\usepackage{cuted} % allows spanning equations over two columns
\usepackage{graphicx}         % to insert graphics
\usepackage{algorithmic}
\usepackage{algorithm}
\usepackage{stfloats}
\usepackage{url}
\usepackage{verbatim}
\usepackage{graphicx}
\usepackage{cite}
\usepackage[top=0.7in, left=0.67in, right=0.67in, bottom=1in]{geometry}

\usepackage{upgreek}			  % to write upstanding greek letters
\usepackage{mathtools}		  % additional math tools
\usepackage{url}					% includes url
\usepackage{enumitem}		% more variability for enumerations
\usepackage{accents}			% more structure in allign environment
\usepackage{bbm}
\usepackage{xcolor}
\usepackage{lipsum}
\usepackage{dirtytalk}
\usepackage{xcolor}
\usepackage{float}
\usepackage{accents}
\usepackage[export]{adjustbox}
\usepackage{setspace}

\usepackage{subfig}

\usepackage{pifont} % nicely looking checkmarks
\usepackage{booktabs, tabularx}
\usepackage{array}

\usepackage{amsthm}
\usepackage{enumitem}

\usepackage{tikz}
\usetikzlibrary{decorations.markings,chains,matrix}
\usetikzlibrary{intersections,calc}
\usepackage{tikz-3dplot} %loading some used tikz libraries
\usepackage{pgfplots}
\usepackage{pgfplotstable}
\pgfplotsset{compat = newest}
\usepackage{datatool}
\usepackage{textcomp}
\usetikzlibrary{shapes,arrows}
\usetikzlibrary{decorations.pathreplacing,angles,quotes}

\usepackage[pscoord]{eso-pic} % placing text boxes
\newtheorem{thm}{Theorem}

\newtheorem{defn}{Definition}

\DeclareMathOperator{\FDP}{FDP}
\DeclareMathOperator{\FDR}{FDR}
\DeclareMathOperator{\TPP}{TPP}
\DeclareMathOperator{\TPR}{TPR}

\DeclareMathOperator{\NHG}{NHG}
\DeclareMathOperator{\Var}{Var}
\DeclareMathOperator{\SNR}{SNR}

\newtheorem{innercustomgeneric}{\customgenericname}
\providecommand{\customgenericname}{}
\newcommand{\newcustomtheorem}[2]{%
  \newenvironment{#1}[1]
  {%
   \renewcommand\customgenericname{#2}%
   \renewcommand\theinnercustomgeneric{##1}%
   \innercustomgeneric
  }
  {\endinnercustomgeneric}
}

\newcustomtheorem{customthm}{Theorem}
\newcustomtheorem{customlemma}{Lemma}
\newcustomtheorem{customcor}{Corollary}
\newcustomtheorem{customassumption}{Assumption}

\newcommand{\y}{\boldsymbol{y}}
\newcommand{\x}{\boldsymbol{x}}

\newcommand{\X}{\boldsymbol{X}}

\newcommand{\bbeta}{\boldsymbol{\beta}}
\newcommand{\bepsilon}{\boldsymbol{\epsilon}}

\newcommand{\A}{\mathcal{A}}

\newcommand{\XK}{\boldsymbol{\protect\accentset{\circ}{X}}}
\newcommand{\xK}{\boldsymbol{\protect\accentset{\circ}{x}}}
\newcommand{\xk}{\protect\accentset{\circ}{x}}
\newcommand{\XWK}{\boldsymbol{\widetilde{X}}}

\newcommand{\C}{\mathcal{C}}

\makeatletter
\let\save@mathaccent\mathaccent
\newcommand*\if@single[3]{%
  \setbox0\hbox{${\mathaccent"0362{#1}}^H$}%
  \setbox2\hbox{${\mathaccent"0362{\kern0pt#1}}^H$}%
  \ifdim\ht0=\ht2 #3\else #2\fi
  }
\newcommand*\rel@kern[1]{\kern#1\dimexpr\macc@kerna}
\newcommand*\widebar[1]{\@ifnextchar^{{\wide@bar{#1}{0}}}{\wide@bar{#1}{1}}}
\newcommand*\wide@bar[2]{\if@single{#1}{\wide@bar@{#1}{#2}{1}}{\wide@bar@{#1}{#2}{2}}}
\newcommand*\wide@bar@[3]{%
  \begingroup
  \def\mathaccent##1##2{%
    \let\mathaccent\save@mathaccent
    \if#32 \let\macc@nucleus\first@char \fi
    \setbox\z@\hbox{$\macc@style{\macc@nucleus}_{}$}%
    \setbox\tw@\hbox{$\macc@style{\macc@nucleus}{}_{}$}%
    \dimen@\wd\tw@
    \advance\dimen@-\wd\z@
    \divide\dimen@ 3
    \@tempdima\wd\tw@
    \advance\@tempdima-\scriptspace
    \divide\@tempdima 10
    \advance\dimen@-\@tempdima
    \ifdim\dimen@>\z@ \dimen@0pt\fi
    \rel@kern{0.6}\kern-\dimen@
    \if#31
      \overline{\rel@kern{-0.6}\kern\dimen@\macc@nucleus\rel@kern{0.4}\kern\dimen@}%
      \advance\dimen@0.4\dimexpr\macc@kerna
      \let\final@kern#2%
      \ifdim\dimen@<\z@ \let\final@kern1\fi
      \if\final@kern1 \kern-\dimen@\fi
    \else
      \overline{\rel@kern{-0.6}\kern\dimen@#1}%
    \fi
  }%
  \macc@depth\@ne
  \let\math@bgroup\@empty \let\math@egroup\macc@set@skewchar
  \mathsurround\z@ \frozen@everymath{\mathgroup\macc@group\relax}%
  \macc@set@skewchar\relax
  \let\mathaccentV\macc@nested@a
  \if#31
    \macc@nested@a\relax111{#1}%
  \else
    \def\gobble@till@marker##1\endmarker{}%
    \futurelet\first@char\gobble@till@marker#1\endmarker
    \ifcat\noexpand\first@char A\else
      \def\first@char{}%
    \fi
    \macc@nested@a\relax111{\first@char}%
  \fi
  \endgroup
}
\makeatother

\DeclareFontFamily{U}{dutchcal}{\skewchar\font=45}
\DeclareFontShape{U}{dutchcal}{m}{n}{<-> s*[1.2] dutchcal-r}{}
\DeclareFontShape{U}{dutchcal}{b}{n}{<-> s*[1.2] dutchcal-b}{}
\DeclareMathAlphabet{\mathdutchcal}{U}{dutchcal}{m}{n}
\SetMathAlphabet{\mathdutchcal}{bold}{U}{dutchcal}{b}{n}
\DeclareMathAlphabet{\mathdutchbcal}{U}{dutchcal}{b}{n}

\newlist{steps}{enumerate}{1}
\setlist[steps, 1]{label = {Step \arabic*:}, ref = {Step \arabic*}}

\newlist{alglist}{enumerate}{1}
\setlist[alglist, 1]{label = {\arabic*.}, ref = {\arabic*}}

\newcommand{\cmark}{\ding{51}}

\newcolumntype{L}[1]{>{\raggedright\arraybackslash}p{#1}}
\newcolumntype{C}[1]{>{\centering\arraybackslash}p{#1}}
\newcolumntype{R}[1]{>{\raggedleft\arraybackslash}p{#1}}

\definecolor{dark_green}{RGB}{102,166,30}
\definecolor{applegreen}{rgb}{0.55, 0.71, 0.0}

\definecolor{dark_red}{RGB}{217,95,2}
\definecolor{bittersweet}{rgb}{1.0, 0.44, 0.37}

\definecolor{dark_yellow}{RGB}{230,171,2}
\definecolor{bananayellow}{rgb}{1.0, 0.88, 0.21}

\newcommand{\placetextbox}[3]{% \placetextbox{<horizontal pos>}{<vertical pos>}{<stuff>}
  \setbox0=\hbox{#3}% Put <stuff> in a box
  \AddToShipoutPictureFG*{% Add <stuff> to current page foreground
    \put(\LenToUnit{#1\paperwidth},\LenToUnit{#2\paperheight}){\vtop{{\null}\makebox[0pt][c]{#3}}}%
  }%
}%

\mathtoolsset{showonlyrefs=true} % Shows only referenced equations
\restylefloat{table}

\def\BibTeX{{\rm B\kern-.05em{\sc i\kern-.025em b}\kern-.08em
    T\kern-.1667em\lower.7ex\hbox{E}\kern-.125emX}}
\begin{document}

\setstretch{0.96}

\title{False Discovery Rate Control for\\ Fast Screening of Large-Scale Genomics Biobanks}

\author{

\IEEEauthorblockN{Jasin Machkour}
\IEEEauthorblockA{
\textit{Technische Universit\"at Darmstadt} \\
64283 Darmstadt, Germany \\
jasin.machkour@tu-darmstadt.de}
\and
\IEEEauthorblockN{Michael Muma}
\IEEEauthorblockA{
\textit{Technische Universit\"at Darmstadt} \\
64283 Darmstadt, Germany \\
michael.muma@tu-darmstadt.de}
\and
\IEEEauthorblockN{Daniel P. Palomar}
\IEEEauthorblockA{
\textit{The Hong Kong University of Science and Technology}\\
Clear Water Bay, Hong Kong SAR, China\\
palomar@ust.hk}

\thanks{The first and second author are supported by the LOEWE initiative (Hesse, Germany) within the emergenCITY center. The second author is also supported by the ERC Starting Grant ScReeningData. The third author is supported by the Hong Kong GRF 16207820 research grant.}
\thanks{Extensive calculations on the Lichtenberg High-Performance Computer of the Technische Universität Darmstadt were conducted for this research.}
}

\maketitle

\begin{abstract}
Genomics biobanks are information treasure troves with thousands of phenotypes (e.g., diseases, traits) and millions of single nucleotide polymorphisms (SNPs). The development of methodologies that provide reproducible discoveries is essential for the understanding of complex diseases and precision drug development. Without statistical reproducibility guarantees, valuable efforts are spent on researching false positives. Therefore, scalable multivariate and high-dimensional false discovery rate (FDR)-controlling variable selection methods are urgently needed, especially, for complex polygenic diseases and traits. In this work, we propose the Screen-T-Rex selector, a fast FDR-controlling method based on the recently developed T-Rex selector. The method is tailored to screening large-scale biobanks and it does not require choosing additional parameters (sparsity parameter, target FDR level, etc). Numerical simulations and a real-world HIV-1 drug resistance example demonstrate that the performance of the Screen-T-Rex selector is superior, and its computation time is multiple orders of magnitude lower compared to current benchmark knockoff methods.
\end{abstract}

\placetextbox{0.5}{0.08}{\fbox{\parbox{\dimexpr\textwidth-2\fboxsep-2\fboxrule\relax}{\footnotesize Published in IEEE Statistical Signal Processing Workshop (SSP), 2-5 July 2023, Hanoi, Vietnam.}}}
\section{Introduction}
\label{sec: Introduction}
The systematic screening of large-scale genomics biobanks enables understanding complex diseases and aids in drug development~\cite{uffelmann2021genome}. Achieving these goals requires finding the few reproducible associations among potentially millions of single nucleotide polymorphisms (SNPs) and a phenotype, i.e., a disease or trait of interest. This allows to further study potentially functionally associated regions on the genome. Large biobanks, such as the UK biobank~\cite{sudlow2015uk}, contain thousands of phenotypes and large ultra-high-dimensional genomics data, where the number of variables (i.e., SNPs) $p$ is much larger than the number of observations $n$. The above described genome-wide association studies (GWAS)~\cite{gwasCatalog} require time and cost intensive follow-up investigations. Hence, it is of utmost importance to keep the number of false discoveries low while discovering as many associations as possible. Therefore, we will consider two metrics:
\begin{enumerate}[label=\arabic*., ref=\arabic*]
\itemsep0em
\item The false discovery rate (FDR) is the expected value of the false discovery proportion (FDP), i.e., the expected percentage of false discoveries among all discoveries: $\FDR \coloneqq \mathbb{E} \big[ \FDP \big] \coloneqq \mathbb{E} \big[ \text{\# False discoveries} / \text{\# Discoveries} \big]$.
\item The true positive rate (TPR) is the expected value of the true positive proportion (TPP), i.e., the expected percentage of true discoveries among all true active variables: $\TPR \coloneqq \mathbb{E} \big[ \TPP \big] \coloneqq \mathbb{E} \big[ \text{\# True discoveries} / \text{\# True actives} \big]$.
\end{enumerate}

Existing FDR-controlling methods allow the user to set a target FDR $\alpha \in [0, 1]$ and select variables such that the FDR is controlled at the target level (i.e., $\alpha$ is not exceeded) while maximizing the number of selected variables and, thus, implicitly maximizing the TPR. Popular methods for low-dimensional settings (i.e., $n \geq p$) are the Benjamini-Hochberg (\textit{BH}) method~\cite{benjamini1995controlling}, the Benjamini-Yekutielli (\textit{BY}) method~\cite{benjamini2001control}, and the more recent \textit{fixed-X} knockoff methods~\cite{barber2015controlling}. Unfortunately, these methods are not applicable for the multivariate analysis of high-dimensional ($p > n$) settings such as GWAS.

In recent years, multivariate FDR-controlling methods for high-dimensional multivariate GWAS have been proposed: \textit{model-X} (and related) knockoff methods~\cite{candes2018panning,barber2018robust,barber2019knockoff,sesia2019gene} and \textit{T-Rex} selector methods~\cite{machkour2021terminating,machkour2022TRexGVS,machkour2022TRexSelector,machkour2022tlars}. However, Figure~1 in~\cite{machkour2021terminating} shows that only the \textit{T-Rex} selector is scalable to millions of variables in a reasonable computation time, while the computation time of the \textit{model-X} methods is multiple orders of magnitude higher and, thus, the method becomes practically infeasible in large-scale settings. Nevertheless, even the comparably low computation time of the \textit{T-Rex} selector for one phenotype might become a burden when conducting GWAS for many phenotypes.

Therefore, we propose the \textit{Screen-T-Rex} selector, a fast version of the \textit{T-Rex} selector. The proposed FDR-controlling method is suitable for conducting large-scale GWAS (with up to millions of SNPs) for tens of thousands of phenotypes. It does not ask the user to set a target FDR level, but provides the user with an estimate of the achieved FDR. In the cases, where the user is not satisfied with the provided FDR estimate, the original \textit{T-Rex} selector should be used with the result of the Screen-T-Rex selector and the desired target FDR as inputs. The proposed \textit{Screen-T-Rex} selector has the following three major innovations/advantages:
\begin{enumerate}[label=\arabic*., ref=\arabic*]
\itemsep0em
\item It provably controls the FDR at the self-estimated level (see Theorems~\ref{theorem: FDR control Screen T-Rex (no bootstrap)} and~\ref{theorem: FDR control Screen T-Rex (with bootstrap)} in Section~\ref{sec: Proposed: Screen-T-Rex Selector}).
\item It does not require the choice of any additional parameters (sparsity parameter, target FDR level, etc.).
\item Its computation time is approximately one order of magnitude lower than that of the original \textit{T-Rex} selector and more than three orders of magnitude lower than that of the \textit{model-X} knockoff methods in our simulations (see Table~\ref{table: results simulated GWAS}).
\end{enumerate}

Organization: Section~\ref{sec: The T-Rex Selector} briefly revisits the original \textit{T-Rex} selector. In Section~\ref{sec: Proposed: Screen-T-Rex Selector}, the \textit{Screen-T-Rex} selector is proposed. Sections~\ref{sec: Numerical Experiments},~\ref{sec: Simulated GWAS}, and~\ref{sec: Real World Example: HIV Data}, compare the proposed method against benchmark methods via numerical simulations, a simulated GWAS, and a real world HIV-1 drug resistance study, respectively. Section~\ref{sec: Conclusion} concludes the paper.

\section{The T-Rex Selector}
\label{sec: The T-Rex Selector}
The \textit{T-Rex} selector~\cite{machkour2021terminating} is a fast FDR-controlling variable selection method for high-dimensional ($p > n$) as well as low-dimensional ($n \geq p$) settings. Figure~\ref{fig: T-Rex selector framework} shows a simplified sketch of the \textit{T-Rex} selector framework. It requires the following inputs:
\begin{enumerate}[label=\arabic*., ref=\arabic*]
\itemsep0em
\item A predictor matrix $\X = [ \x_{1} \, \cdots \, \x_{p} ]$ that, e.g., contains $p$ SNPs $\x_{1}, \ldots, \x_{p}$ as columns, where $\x_{j} = [ x_{1j} \, \cdots \, x_{nj} ]^{\top}$ contains $n$ observations of the $j$th SNP. That is, the $i$th row of $\X$ contains the measurements of all SNPs for the $i$th subject.
\item A response vector $\y = [ y_{1} \, \cdots \, y_{n} ]^{\top}$ that, e.g., contains the phenotypes of all $n$ subjects. These can be measurements of the disease progression or, in a simple case-control study, the value ``$1$'' for cases and the value ``$0$'' for controls.
\item The target FDR level $\alpha \in [0, 1]$.
\end{enumerate}
%%%%%%%%%%%%%%%%
%%%%%%%%%%%%%%%%
%%%%%%%%%%%%%%%%
\begin{figure}[t]
\begin{center}
\scalebox{0.625}{
\begin{tikzpicture}[>=stealth]

  %coordinates for nodes
  \coordinate (orig)   at (0,0);
  \coordinate (sample)   at (1,0.5);
  \coordinate (merge)   at (3,0.5);
  \coordinate (varSelect)   at (5,0.5);
  \coordinate (tFDR)   at (15.5,-1.1);
  \coordinate (fuse)   at (7.5,0.5);
  \coordinate (output)   at (9.5,0.5);
  
  %coordinates for bending points of arrows
  \coordinate (between_scale_rank)   at (0.5,0.31);
  \coordinate (X_prime_to_tFDR_point)   at (0.5,6);
  \coordinate (X_prime_to_tFDR_point_point)   at (9,6);
  \coordinate (X_prime_to_merge_point)   at (0.5,-2.5);
  \coordinate (center_to_tFDR_point)   at (5.00,3.7);
  \coordinate (tFDR_to_fuse_point)   at (7.5,3.7);
  \coordinate (tFDR_to_sample_point)   at (4,5);

  %coordinates for starting points of arrows
  \coordinate (Arrow_N_GenDummy)   at (1,3.06);
  \coordinate (Arrow_X_indVar)   at (3,3.06);
  \coordinate (Arrow_targetFDR_tFDR)   at (10,5.9);
  
  %coordinates for end points of arrows
   \coordinate (inference_Arrow)   at (15,-1.06);
   \coordinate (fuse_Arrow)   at (16.1,2.06);
  
  %coordinates for \vdots
  \coordinate (vdots1)   at (2.0,0.4);
  \coordinate (vdots2)   at (4.0,0.4);
  \coordinate (vdots3)   at (6.25,0.4);
  
  %coordinates for output of voting step
  \coordinate (fuse_node)   at (15.00,0.5);
  
  %nodes
  \node[draw, minimum width=.7cm, minimum height=4cm, anchor=center , align=center] (C) at (sample) {\rotatebox{90}{\Large Generate Dummies}};
  \node[draw, minimum width=.7cm, minimum height=4cm, anchor=center, align=center] (D) at (merge) {\rotatebox{90}{\Large Append}};   
  \node[draw, minimum width=.7cm, minimum height=5.5cm, anchor=center, align=center] (E) at (varSelect) {\rotatebox{90}{\Large Forward Variable Selection}};
  \node[draw, minimum width=.7cm, minimum height=5.5cm, anchor=center, align=center] (H) at (fuse) {\rotatebox{90}{\Large Calibrate \& Fuse}};
  \node[draw, minimum width=2.5cm, minimum height=.7cm, anchor=center, align=center] (N) at (output) {\Large Output: \\[0.3em] \Large $\widehat{\mathcal{A}}_{L}(v^{*}, T^{*})$};
  \node (J) at (vdots1) {\Large $\vdots$};
  \node (K) at (vdots2) {\Large $\vdots$};
  \node (L) at (vdots3) {\Large $\vdots$};
  
  %edges
  \draw[->] (Arrow_N_GenDummy) -- node[above, pos = 0.1]{\large $\sim\mathcal{N}(0, 1)$} ($(C.90)$); 
  \draw[->] (Arrow_X_indVar) -- node[above, pos = 0.1]{\large $\X$} ($(D.90)$); 
     
  \draw[->] ($(C.0) + (0,1.5)$) -- node[above]{\large $\XK_{1}$} ($(D.0) + (-0.7,1.5)$);
  \draw[->] ($(C.0) + (0,0.75)$) -- node[above]{\large $\XK_{2}$} ($(D.0) + (-0.7,0.75)$);
  \draw[->] ($(C.0) + (0,-1.5)$) -- node[above]{\large $\XK_{K}$} ($(D.0) + (-0.7,-1.5)$);
     
  \draw[->] ($(D.0) + (0,1.5)$) -- node[above]{\large $\XWK_{1}$} ($(E.0) + (-0.7,1.5)$);
  \draw[->] ($(D.0) + (0,0.75)$) -- node[above]{\large $\XWK_{2}$} ($(E.0) + (-0.7,0.75)$);
  \draw[->] ($(D.0) + (0,-1.5)$) -- node[above]{\large $\XWK_{K}$} ($(E.0) + (-0.7,-1.5)$);
     
  \draw[->] ($(E.0) + (0,1.5)$) -- node[above]{\large $\C_{1, L}(T)$} ($(H.0) + (-0.7,1.5)$);
  \draw[->] ($(E.0) + (0,0.75)$) -- node[above]{\large $\C_{2, L}(T)$} ($(H.0) + (-0.7,0.75)$);
  \draw[->] ($(E.0) + (0,-1.5)$) -- node[above]{\large $\C_{K, L}(T)$} ($(H.0) +  (-0.7,-1.5)$);
     
  \draw[->] (center_to_tFDR_point) -- node[above, pos = 0.1]{\large $\y$} ($(E.90)$); 
  \draw[->] (tFDR_to_fuse_point) -- node[above, pos = 0.1]{\Large $\alpha$} ($(H.90)$);
 
  \coordinate (between_varSelect_fuse1)   at ($(E.0) + (2.75,1.5)$);
  \coordinate (between_varSelect_fuse2)   at ($(E.0) + (2.75,0.75)$);
  \coordinate (between_varSelect_fuse3)   at ($(E.0) + (2.75,-1.5)$);
 
  \draw[->] (H) -- (N);
\end{tikzpicture}}
\end{center}
\setlength{\belowcaptionskip}{-8pt}
\caption{Sketch of the \textit{T-Rex} selector framework~\cite{machkour2022TRexGVS}.}
\label{fig: T-Rex selector framework}
\end{figure}
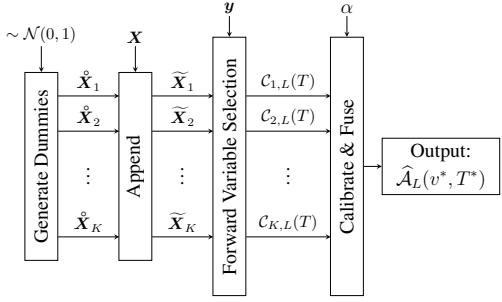

%%%%%%%%%%%%%%%%
%%%%%%%%%%%%%%%%
%%%%%%%%%%%%%%%%
First, the \textit{T-Rex} selector generates $K$ dummy predictor matrices $\XK_{k} = [\xK_{1} \, \cdots \, \xK_{L}]$, $k = 1, \ldots, K$ (containing $L$ dummy predictors), where each element of the $l$th dummy predictor $\xK_{l} = [ \xk_{1l} \, \cdots \, \xk_{nl} ]^{\top}$ is sampled from a univariate standard normal distribution. After appending each dummy matrix to the original predictor matrix, $K$ independent random experiments are conducted by applying a forward variable selection algorithm to each extended predictor matrix $\XWK_{k} = [ \X \,\, \XK_{k} ]$, $k = 1, \ldots, K$, where the response $\y$ acts as the supervising vector. The forward variable selection algorithm includes one variable at a time and terminates after $T \geq 1$ dummies have been included. In~\cite{machkour2021terminating}, it is proposed to use the LARS algorithm~\cite{efron2004least} (or related methods, see, e.g.,  \cite{tibshirani1996regression,zou2005regularization,zou2006adaptive}), which assumes a linear relationship between the predictors and the response, as the forward selector within the \textit{T-Rex} framework. Following the notation of the \textit{T-Rex} selector, the linear model is defined by
\begin{equation}
\y = \X \bbeta + \bepsilon,
\label{eq: linear model}
\end{equation}
where $\bbeta$ is the sparse coefficient vector and $\bepsilon \sim \mathcal{N}(\boldsymbol{0}, \sigma^{2} \boldsymbol{I})$ is the Gaussian noise vector. The obtained candidate sets $\C_{k, L}(T)$, $k = 1, \ldots, K$, contain the included original variables (after removal of the $T$ dummy variables). In the ``Calibrate \& Fuse'' step, the relative occurrence of each original variable $j$ in the candidate sets, i.e., $\Phi_{T, L}(j) \in [0, 1]$, $j = 1, \ldots, p$, is computed. The \textit{T-Rex} calibration algorithm automatically determines
\begin{enumerate}[label=\arabic*., ref=\arabic*]
\itemsep0em
\item the number of dummies $L$,
\item the optimal number of included dummies before terminating the random experiments $T^{*}$, and
\item the optimal voting level $v^{*} \in [ 0, 1 )$
\end{enumerate}
to obtain the set of selected variables
\begin{equation}
\widehat{\A}_{L}(v^{*}, T^{*}) \coloneqq \lbrace j : \Phi_{T^{*}, L}(j) > v^{*} \rbrace
\label{eq: selected active set}
\end{equation}
such that the FDR is controlled at the target level $\alpha$, while the number of selected variables $\big| \widehat{\A}_{L}(v^{*}, T^{*}) \big|$ is maximized. Note that in accordance with the suggestion in~\cite{machkour2021terminating}, we have conducted $K = 20$ random experiments throughout this work.

\begin{defn}[FDP and FDR]
Let $R_{T, L}(v) \coloneqq \big| \widehat{\A}_{L}(v, T) \big|$ and $V_{T, L}(v) \coloneqq \big| \lbrace \text{null } j : \Phi_{T, L}(j) > v \rbrace \big|$ be the number of selected variables and the number of selected null variables (i.e., false positives), respectively. Define $a \lor b \coloneqq \max \lbrace a, b \rbrace$, $a,b \in \mathbb{R}$. Then, the FDR and the false discovery proportion (FDP) are defined by
\begin{equation}
\FDR(v, T, L) \coloneqq \mathbb{E} \big[ \FDP(v, T, L) \big] \coloneqq \mathbb{E}\bigg[ \dfrac{V_{T, L}(v)}{R_{T, L}(v) \lor 1} \bigg].
\label{eq: definition - FDR and FDP with T-Rex selector notation}
\end{equation}
\label{definition: FDR and FDP with T-Rex selector notation}
\end{defn}

\section{Proposed: Screen-T-Rex Selector}
\label{sec: Proposed: Screen-T-Rex Selector}
Two versions of the \textit{Screen-T-Rex} selector are proposed, the corresponding FDR control theorems are presented, and an algorithm for screening genomics biobanks is formulated.

\subsection{Ordinary Screen-T-Rex Selector}
\label{subsec: Ordinary Screen-T-Rex Selector}
While the original \textit{T-Rex} selector determines $T$, $L$, and $v$ such that the FDR is controlled at the user-defined target level, the \textit{Screen-T-Rex} selector fixes $(T, L, v) = (1, p, 0.5)$. This is a special case of the original \textit{T-Rex} selector that
\begin{enumerate}[label=\arabic*., ref=\arabic*]
\itemsep0em
\item is harnessed by the proposed \textit{Screen-T-Rex} selector to determine an estimator of the FDR and
\item requires a much lower computation time than the original \textit{T-Rex} selector and other benchmark methods (see Table~\ref{table: results simulated GWAS} in Section~\ref{sec: Simulated GWAS}).
\end{enumerate}

The FDR estimator of the proposed ordinary \textit{Screen-T-Rex} selector is given by
\begin{equation}
\widehat{\alpha} \coloneqq 1 / (R_{1, p}(0.5) \lor 1),
\label{eq: FDR estimator of Screen-T-Rex selector}
\end{equation}
i.e., one divided by the number of selected variables. The intuition behind this estimator is as follows: $T = 1$ dummy variable is allowed to enter the solution paths of the random experiments before terminating the forward selection processes. So, in each random experiment one out of $p$ dummies is included. Therefore, we expect, on average, no more than one out of at most $p$ null variables to be included in each candidate set $\C_{k, L}(T)$, and, consequently, no more than one null variable among all selected variables. This idea is formalized in the following FDR control result:
\begin{thm}[FDR control - ordinary \textit{Screen-T-Rex}]
Define $\widehat{\alpha} \coloneqq 1 / (R_{1, p}(0.5) \lor 1)$. Then, $\FDR = \mathbb{E}[ \FDP ] \leq \widehat{\alpha}$, i.e., the FDR is controlled at the estimated level $\widehat{\alpha}$.\footnote{For simplicity, $\FDR \coloneqq \FDR(1, p, 0.5)$ and $\FDP \coloneqq \FDP(1, p, 0.5)$.}
\label{theorem: FDR control Screen T-Rex (no bootstrap)}
\end{thm}
\begin{proof}
With Definition~\ref{definition: FDR and FDP with T-Rex selector notation}, we obtain
\begin{equation}
\FDP = \dfrac{V_{1, p}(0.5)}{R_{1, p}(0.5) \lor 1} = \widehat{\alpha} \cdot V_{1, p}(0.5).
\label{eq: FDP in proof of theorem - FDR control Screen T-Rex (no bootstrap)}
\end{equation}

Let $p = p_{1} + p_{0}$, where $p_{1}$ and $p_{0}$ are the number of true active and null variables, respectively. Taking the expectation of~\eqref{eq: FDP in proof of theorem - FDR control Screen T-Rex (no bootstrap)} yields
\begin{equation}
\FDR = \mathbb{E}[\FDP] 
= 
\widehat{\alpha} \cdot \mathbb{E} \big[ V_{1, p}(0.5) \big]
\leq 
\widehat{\alpha} \cdot \dfrac{p_{0}}{p + 1} \leq \widehat{\alpha},
\end{equation}
where the first inequality follows from $V_{1, p}(0.5)$ being stochastically dominated by the negative hypergeometric distribution $\NHG(p_{0} + p, p_{0}, 1)$, whose expected value is given by $p_{0} / (p + 1)$ (for details on the $\NHG$, see~\cite{machkour2021terminating}).
\label{proof: theorem - FDR control Screen T-Rex (no bootstrap)}
\end{proof}

\subsection{Confidence-Based Screen-T-Rex Selector}
\label{subsec: Confidence Interval Based Screen-T-Rex Selector}
The above proposed ordinary \textit{Screen-T-Rex} selector, as well as the original \textit{T-Rex} selector, only considers the relative occurrences of the candidate variables in the selected active sets $\C_{k, L}(T)$ and disregards the original and dummy coefficient estimates, i.e.,
\begin{enumerate}[label=\arabic*., ref=\arabic*]
\itemsep0em
\item $\hat{\beta}_{j, k}(T, L)$, $j = 1, \ldots, p$, (i.e., coefficient estimate of the $j$th original variable in the $k$th random experiment and
\item $\hat{\beta}_{l, k}^{\circ}(T, L)$, $l = 1, \ldots, L$, (i.e., coefficient estimate of the $l$th dummy variable in the $k$th random experiment.
\end{enumerate}
However, since the dummy variables act as flagged null variables (for details, see \cite{machkour2021terminating}), the coefficients of the dummies contain information about the distribution of the coefficients of the null variables. Therefore, we propose to harness the coefficient estimates of the dummies to construct a confidence interval
\begin{equation}
C(\gamma) \coloneqq [ c_{1}(\gamma), c_{2}(\gamma) ], \quad \gamma \in [0, 1],
\label{eq: confidence interval using dummy coefficient estimates}
\end{equation} 
where $c_{1}(\gamma)$ and $c_{2}(\gamma)$ are the lower and upper bound, respectively, and $\gamma$ is the confidence level. The coefficient estimates of the null variables can also be expected to lie within the same confidence interval. Therefore, instead of selecting variables based on their relative occurrences, we replace $V_{1, p}(0.5)$ and $R_{1, p}(0.5)$ in Definition~\ref{definition: FDR and FDP with T-Rex selector notation} and Theorem~\ref{theorem: FDR control Screen T-Rex (no bootstrap)} by
\begin{align}
& V_{1, p}^{(C)}(\gamma) \coloneqq \big| \big\lbrace \text{null } j : \widebar{\hat{\beta}}_{j}(1, p) \notin C(\gamma) \big\rbrace \big| \text{ and}
\label{eq: number of selected null variables (with bootstrap)}
\\
& R_{1, p}^{(C)}(\gamma) \coloneqq \big| \widehat{\A}_{p}^{(C)}(\gamma, 1) \big| \coloneqq \big| \big\lbrace j : \widebar{\hat{\beta}}_{j}(1, p) \notin C(\gamma) \big\rbrace \big|,
\label{eq: number of selected variables (with bootstrap)}
\end{align}
respectively, where $\widebar{\hat{\beta}}_{j}(1, p) \coloneqq \frac{1}{K} \sum_{k = 1}^{K} \hat{\beta}_{j, k}(1, p)$. That is, only candidate variables whose averaged (over $K$ random experiments) coefficient estimates are not inside the confidence interval $C(\gamma)$ are selected.

We propose to construct the confidence interval in~\eqref{eq: confidence interval using dummy coefficient estimates} using the non-parametric bootstrap with $1{,}000$ resamples of the vector containing the $K = 20$ non-zero dummy coefficient estimates. Since, in all our simulations, the distribution of the bootstrapped standard errors of the averaged non-zero dummy coefficient estimates followed the standard normal distribution, we construct a normal bootstrap confidence interval (for details, see~\cite{efron1994introduction,davison1997bootstrap}). In the following theorem, we state how the most liberal confidence level $\gamma$ can be determined such that the FDR is controlled at the estimated level by the confidence-based \textit{Screen-T-Rex} selector:
\begin{thm}[FDR control - confidence-based \textit{Screen-T-Rex}]
Define $\gamma \coloneqq \inf \big\lbrace \gamma^{\prime} \in [0, 1] : R_{1, p}^{(C)}(\gamma^{\prime}) \leq R_{1, p}(0.5) \big\rbrace$ and $\widehat{\alpha}_{C} \coloneqq 1 / (R_{1, p}^{(C)}(\gamma) \lor 1)$. Suppose that $V_{1, p}^{(C)}(\gamma) \overset{d}{\leq} V_{1, p}(0.5)$, where $\overset{d}{\leq}$ denotes stochastic dominance. Then, $\FDR = \mathbb{E}[ \FDP ] \leq \widehat{\alpha}_{C}$.
\label{theorem: FDR control Screen T-Rex (with bootstrap)}
\end{thm}
\begin{proof}
With Definition~\ref{definition: FDR and FDP with T-Rex selector notation} and Equations~\eqref{eq: number of selected null variables (with bootstrap)} and~\eqref{eq: number of selected variables (with bootstrap)}, we obtain
\begingroup
\allowdisplaybreaks
\begin{align}
\FDR &= \mathbb{E} [ \FDP ] 
= \mathbb{E} \Bigg[ \dfrac{V_{1, p}^{(C)}(\gamma)}{R_{1, p}^{(C)}(\gamma) \lor 1} \Bigg]
= \widehat{\alpha}_{C} \cdot \mathbb{E} \Big[ V_{1, p}^{(C)}(\gamma) \Big] \\
&\leq \widehat{\alpha}_{C} \cdot \mathbb{E} \big[ V_{1, p}(0.5) \big] \leq \widehat{\alpha}_{C},
\end{align}
\endgroup
where the first inequality follows from $V_{1, p}^{(C)}(\gamma) \overset{d}{\leq} V_{1, p}(0.5)$ and the second inequality is the same as in the proof of Theorem~\ref{theorem: FDR control Screen T-Rex (no bootstrap)}.
\label{proof: theorem - FDR control Screen T-Rex (with bootstrap)}
\end{proof}

\subsection{Screening Genomics Biobanks}
\label{subsec: Screening Genomics Biobanks}
The \textit{Screen-T-Rex} selector is intended to be used for screening thousands of phenotypes in large biobanks, while only using the original \textit{T-Rex} selector in the cases where the estimated FDR is not acceptable to the user. Here, the user sets the target FDR for the original \textit{T-Rex} selector and a lower and upper bound $\alpha_{l}$ and $\alpha_{u}$, respectively, for the estimated FDRs by both versions of the \textit{Screen-T-Rex} selector. The lower bound is required to avoid solutions at very low estimated FDRs, since these would yield a low power (i.e., TPR). Algorithm~\ref{algorithm: Screening genomics biobanks} summarizes the proposed work flow.
%
%%%%%%%%%%%%%%%%%%%%%%%%%%%%
%%%%%%%%%%%%%%%%%%%%%%%%%%%%
%%%%%%%%%%%%%%%%%%%%%%%%%%%%
\begin{table*}[b]
\centering
\begin{tabular}{@{}L{\dimexpr0.26\linewidth-2\tabcolsep\relax}@{}C{\dimexpr0.095\linewidth-2\tabcolsep\relax}@{}C{\dimexpr0.11\linewidth-2\tabcolsep\relax}@{}C{\dimexpr0.18\linewidth-2\tabcolsep\relax}@{}C{\dimexpr0.105\linewidth-2\tabcolsep\relax}@{}C{\dimexpr0.2\linewidth-2\tabcolsep\relax}@{}C{\dimexpr0.2\linewidth-2\tabcolsep\relax}}
%\toprule
\bfseries Methods
& \bfseries FDR control?
& \bfseries Av. FDP (in~$\%$)
& \bfseries Av. estimated/target FDR (in~$\%$)
& \bfseries Av. TPP (in~$\%$)
& \bfseries Av. sequential comp. time (hh:mm:ss)
& \bfseries Av. relative sequen- tial comp. time \\
\midrule
		\textbf{Proposed:} \\[0.0em]
		\textit{1. Screen-T-Rex (ordinary)} & \cmark & $15.96$ & $18.57$ & $47.2$ & $\boldsymbol{00$:$00$:$44}$ & $\boldsymbol{1}$ \\[0.0em]
		\textit{2. Screen-T-Rex (conf.-based)} & \cmark & $10.16$ & $12.5$ & $31.7$ & $00$:$00$:$45$ & $1.02$ \\[0.0em]
\midrule
		\textbf{Benchmarks:} \\[0.0em]
		\textit{3. T-Rex} & \cmark & $6.45$ & $10$ & $38.5$ & $00$:$06$:$39$ & $8.88$ \\[0.0em]
		\textit{4. model-X+} & \cmark & $0$ & $10$ & $0$ & $20$:$00$:$38$ & $1601.39$
%\bottomrule
\end{tabular}
\caption{For all methods, the average achieved FDP is lower than the average estimated/target FDR, i.e., all methods control the FDR. The sequential computation time (of both versions) of the proposed \textit{Screen-T-Rex} selector is more than three orders of magnitude lower than that of the \textit{model-X} knockoff+ method. Nearly one order of magnitude is gained compared to the original \textit{T-Rex} selector. Applying Algorithm~\ref{algorithm: Screening genomics biobanks} with $\alpha = 10\%$, $\alpha_{l} = 5\%$, and $\alpha_{u} = 20\%$, yields an average FDP and TPP of $15.96\%$ and $47.2\%$, respectively, without requiring Step~\ref{algorithm: Screening genomics biobanks Step 2.3}.3.
}
\label{table: results simulated GWAS}
\end{table*}
%%%%%%%%%%%%%%%%%%%%%%%%%%%%
%%%%%%%%%%%%%%%%%%%%%%%%%%%%
%%%%%%%%%%%%%%%%%%%%%%%%%%%%
%

\section{Numerical Experiments}
\label{sec: Numerical Experiments}
We simulate a high-dimensional data setting according to the linear model in~\eqref{eq: linear model} with $n = 300$ samples and $p = 1{,}000$ predictors (i.e, candidate variables), and $p_{1} = 10$ true active variables. The noise variance $\sigma^{2}$ is chosen such that the signal-to-noise-ratio $\SNR \coloneq \Var(\X\bbeta) / \Var(\bepsilon)$ (for details, see~\cite{machkour2022TRexSelector}) takes on the values on the $x$-axes in Figure~\ref{fig: sweep p1 and sweep snr plots p = 1000}. Note that the FDP and TPP in Figure~\ref{fig: sweep p1 and sweep snr plots p = 1000} are averaged over $955$ Monte Carlo replications, respectively, and, therefore, are estimates of the FDR and TPR, respectively. A discussion of the results is provided within the caption of Figure~\ref{fig: sweep p1 and sweep snr plots p = 1000}.
\begin{algorithm}[h]
\caption{Screening Genomics Biobanks.}
\begin{alglist}
\item \textbf{Input}: $\alpha$, $\alpha_{l}$, and $\alpha_{u}$.
\label{algorithm: Screening genomics biobanks Step 1}
\item For each considered phenotype in the biobank do:
\label{algorithm: Screening genomics biobanks Step 2}
\begin{enumerate}
\item[2.1.] \textbf{Run} the \textit{Screen-T-Rex} selector and obtain the estimated FDR levels $\widehat{\alpha}$ and $\widehat{\alpha}_{C}$.
\label{algorithm: Screening genomics biobanks Step 2.1}
\item[2.2.] \textbf{Determine} the final set of selected variables $\widehat{\A}$ as follows:
\begin{equation}
\widehat{\A} \coloneqq
\begin{cases}
\widehat{\A}_{p}^{(C)}(\gamma, 1), 
\quad  
\textstyle
\begin{array}{l}
\alpha_{l} \leq \widehat{\alpha}_{C} \leq \alpha_{u} \text{ \& }\\ 
\max \lbrace \widehat{\alpha}_{C}, \widehat{\alpha} \cdot I(\widehat{\alpha} \leq \alpha_{u}) \rbrace = \widehat{\alpha}_{C}
\end{array}
\\[1em]
\widehat{\A}_{p}(0.5, 1), 
\quad   \,
\textstyle
\begin{array}{l}
\alpha_{l} \leq \widehat{\alpha} \leq \alpha_{u} \text{ \& }\\ 
\max \lbrace \widehat{\alpha}_{C} \cdot I(\widehat{\alpha}_{C} \leq \alpha_{u}), \widehat{\alpha} \rbrace = \widehat{\alpha}
\end{array}
\\[0.8em]
\varnothing, \qquad \qquad \quad \,\, \mathrm{otherwise}
\end{cases} \hspace{-0.6cm},
\end{equation}
where $I(a \leq b)$, $a, b \in \mathbb{R}$, is the indicator function that has the value one if $a \leq b$ and zero otherwise. $\varnothing$ denotes the empty set. Convention: If $\widehat{\alpha} = \widehat{\alpha}_{C}$ and all conditions in the first two cases are satisfied, then $\widehat{\A} \coloneqq \widehat{\A}_{p}^{(C)}(\gamma, 1)$.
\label{algorithm: Screening genomics biobanks Step 2.2}
\item[2.3.] If $\widehat{\A} = \varnothing$, \textbf{run} the \textit{T-Rex} selector with target FDR $\alpha$ and \textbf{determine}
\begin{equation}
\widehat{\A} \coloneqq \widehat{\A}_{L}(v^{*}, T^{*}).
\end{equation}
\label{algorithm: Screening genomics biobanks Step 2.3}
\vspace{-1pc}
\end{enumerate}
\item \textbf{Output}: Selected active set $\widehat{\A}$.
\label{algorithm: Screening genomics biobanks Step 3}
\end{alglist}
\label{algorithm: Screening genomics biobanks}
\end{algorithm}
%\vspace{-0.95em}
%
%%%%%%%%%%%%%%%%%%%%%%%%%%%%%%%%%%%%
%%%%%%%%%%%%%%%%%%%%%%%%%%%%%%%%%%%%
%%%%%%%%%%%%%%%%%%%%%%%%%%%%%%%%%%%%
\begin{figure}[t]
  \centering
	\vspace*{-0.135cm}
	\hspace*{-0.8em}
  \subfloat[]{
  		\scalebox{1}{
  			\includegraphics[width=0.48\linewidth]{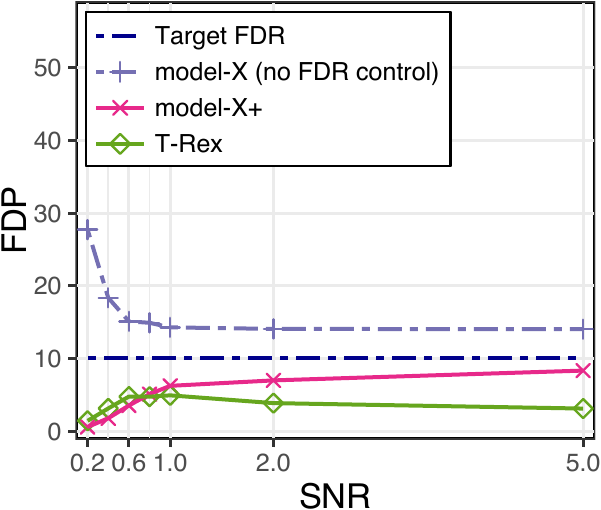}
  		}
   		\label{fig: FDP_vs_SNR_p_1000_Optimal_T_L}
   }
	\hspace*{-1.2em}
  \subfloat[]{
  		\scalebox{1}{
  			\includegraphics[width=0.49\linewidth]{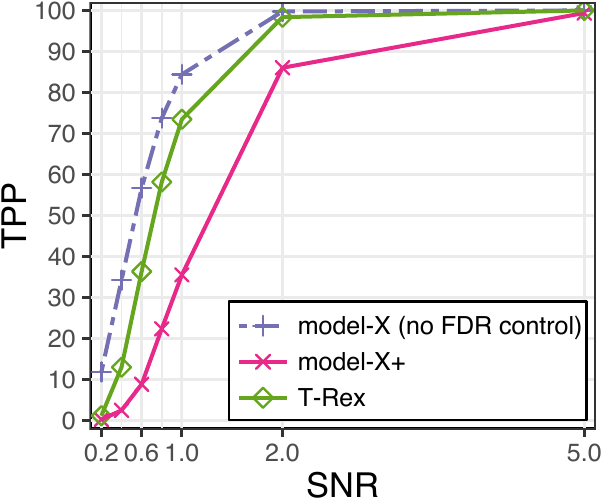}
  		}
   		\label{fig: TPP_vs_SNR_p_1000_Optimal_T_L}
   }
   \\
   	\hspace*{-0.8em}
  \subfloat[]{
  		\scalebox{1}{
  			\includegraphics[width=0.48\linewidth]{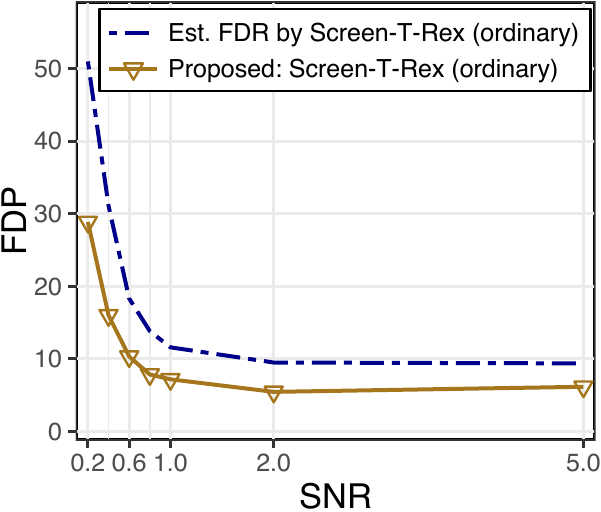}
  		}
   		\label{fig: FDP_vs_SNR_p_1000_Screen_T_Rex_ordinary}
   }
	\hspace*{-1.2em}
  \subfloat[]{
  		\scalebox{1}{
  			\includegraphics[width=0.49\linewidth]{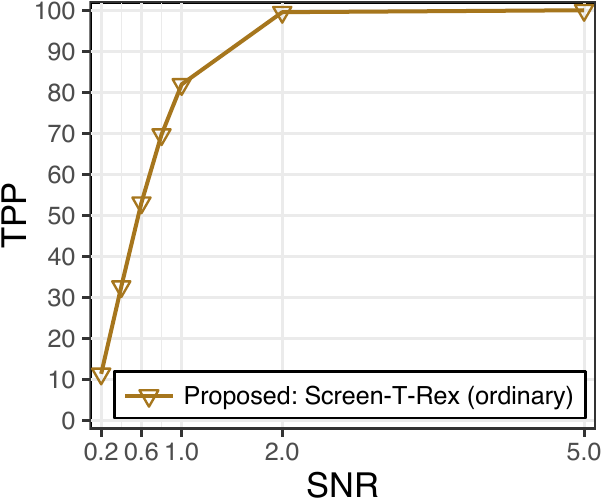}
  		}
   		\label{fig: TPP_vs_SNR_p_1000_Screen_T_Rex_ordinary}
   }
   \\
   	\hspace*{-0.8em}
  \subfloat[]{
  		\scalebox{1}{
  			\includegraphics[width=0.48\linewidth]{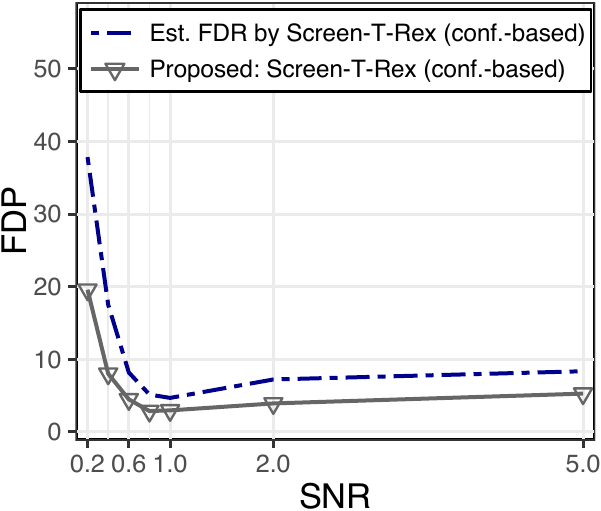}
  		}
   		\label{fig: FDP_vs_SNR_p_1000_Screen_T_Rex_confidence_based}
   }
	\hspace*{-1.2em}
  \subfloat[]{
  		\scalebox{1}{
  			\includegraphics[width=0.49\linewidth]{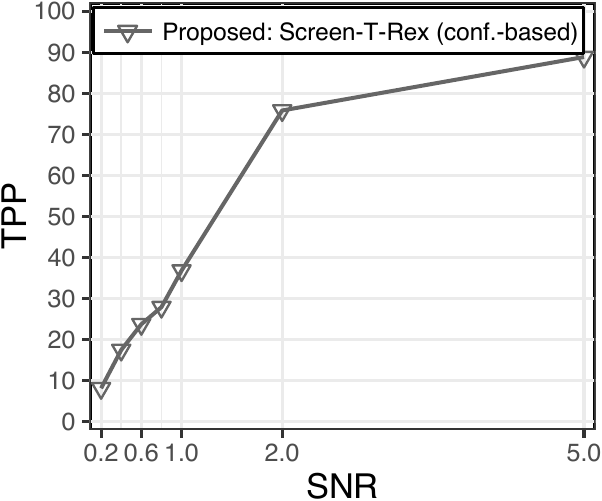}
  		}
   		\label{fig: TPP_vs_SNR_p_1000_Screen_T_Rex_confidence_based}
   }
\setlength{\belowcaptionskip}{-8pt}
  \caption{Figures~(a) and~(b)~(see also \cite{machkour2021terminating}) display the results of the benchmark methods. We observe that the \textit{T-Rex} selector and the \textit{model-X} knockoff+ method control the FDR at the target level of $10$\%, while the \textit{model-X} knockoff method fails to control the FDR. Figures~(c)~-~(f) show that the proposed ordinary and the confidence-based \textit{Screen-T-Rex} selector both control the FDR at the self-estimated levels while achieving a reasonably high TPR. The confidence-based version is capable of controlling the FDR at lower levels than the ordinary version but, in turn, achieves a lower TPR.}
  \label{fig: sweep p1 and sweep snr plots p = 1000}
\end{figure}
%%%%%%%%%%%%%%%%%%%%%%%%%%%%%%%%%%%%
%%%%%%%%%%%%%%%%%%%%%%%%%%%%%%%%%%%%
%%%%%%%%%%%%%%%%%%%%%%%%%%%%%%%%%%%%
%

\section{Simulated GWAS}
\label{sec: Simulated GWAS}
The simulation setup and preprocessing of the data in this section is the same as the one in~\cite{machkour2021terminating}. That is, $100$ genomics data sets are simulated using the software HAPGEN2~\cite{su2011hapgen2}. It takes haplotypes from the HapMap $3$ project~\cite{international2010integrating} as an input and generates the predictor matrix $\X$ that contains $p = 20{,}000$ SNPs/columns and $n = 1{,}000$ observations/rows (i.e., $700$ cases and $300$ controls). The phenotype vector $\y$ contains ones for cases and zeros for controls. Each of the $100$ simulated data sets contains $p_{1} = 10$ true active SNPs, i.e., SNPs that are associated with the phenotype. The results in Table~\ref{table: results simulated GWAS} are averaged over all data sets. A discussion of the results is provided within the caption of Table~\ref{table: results simulated GWAS}.\footnote{Note that only the relative but not the absolute computation times are representative for similar settings, since, due to the energy crisis in Europe, the CPUs of the Lichtenberg High-Performance Computer of the Technische Universität Darmstadt operate at a reduced clock frequency.}

\section{Real World Example: HIV Data}
\label{sec: Real World Example: HIV Data}
Many antiretroviral drugs are used in HIV-$1$ infection therapies. However, mutations may decrease the susceptibility to some drugs and, thus, lead to an increased drug resistance of the virus. Therefore, it is desired to detect mutations associated with resistance against all existing drugs to determine which drugs to use for treating HIV-$1$ and to develop new drugs to which mutated HIV-$1$ viruses are highly susceptible. In order to also compare the proposed methods against classical methods for the low-dimensional setting, we consider a low-dimensional benchmark HIV-$1$ data set that was described and analyzed in \cite{rhee2005hiv, rhee2006genotypic} and served as a benchmark data set for the existing \textit{fixed-X} knockoff method \cite{barber2015controlling}. It can be downloaded from a Stanford University database.\footnote{URL (last access: $31$st January $2023$):\\ \url{https://hivdb.stanford.edu/pages/published_analysis/genophenoPNAS2006/}}
The performance of the proposed \textit{Screen-T-Rex} selector and the benchmark methods in detecting the mutations that are associated with HIV-$1$ drug resistance for individual \textit{protease inhibitor} (PI)-type drugs is assessed. The same setup as in \cite{barber2015controlling} is used, i.e., the same preprocessing steps are applied and the same benchmark mutation positions from treatment-selected mutation (TSM) lists in \cite{rhee2005hiv} are used. Table~\ref{table: Overview of the HIV-1 dataset (PI)} and Figure~\ref{fig: HIV_PI_drugs} display and discuss the results.
%
%%%%%%%%%%%%%%%%%%%%%%
%%%%%%%%%%%%%%%%%%%%%%
%%%%%%%%%%%%%%%%%%%%%%
\begin{table}[t]
\caption{Results for the HIV-$1$ PI-type drugs: Applying Algorithm~\ref{algorithm: Screening genomics biobanks} with $\alpha = 3\%$, $\alpha_{l} = 2\%$, and $\alpha_{u} = 4\%$, the selected active set of the original \textit{T-Rex} selector is only picked for ATV. For APV and IDV, the results of the confidence-based \textit{Screen-T-Rex} selector are picked. For the remaining drugs, the results of the ordinary \textit{Screen-T-Rex} selector are picked (compare results in Figure~\ref{fig: HIV_PI_drugs}).}
\centering
\begin{tabular}{
@{}L{\dimexpr0.18\textwidth-2\tabcolsep\relax}
@{}C{\dimexpr0.06\textwidth-2\tabcolsep\relax}
@{}C{\dimexpr0.06\textwidth-2\tabcolsep\relax}
@{}C{\dimexpr0.08\textwidth-2\tabcolsep\relax}
@{}C{\dimexpr0.11\textwidth-2\tabcolsep\relax}
@{}R{\dimexpr0.13\textwidth-2\tabcolsep\relax}
}
%\toprule
\bfseries Drug
& $n$
& $p$
& Target FDR
& Est. FDR (ordinary)
& Est. FDR (conf.-based)
\\
\midrule

		Amprenavir 	(\textbf{APV})  & 767 & 201 & $3$ \% & $3.57$ \% & $3.70$ \% \\
  		Atazanavir		(\textbf{ATV})  & 328 & 147 & $3$ \% & $4.76$ \% & $0.00$ \% \\
   		Indinavir			(\textbf{IDV})   & 825 & 206 & $3$ \% & $3.33$ \% & $3.33$ \% \\
   		Lopinavir			(\textbf{LPV})  & 515 & 184 & $3$ \% & $3.85$ \% & $0.00$ \% \\
        Nelfinavir			(\textbf{NFV})  & 842 & 207 & $3$ \% & $3.70$ \% & $0.00$ \% \\
        Ritonavir			(\textbf{RTV})  & 793 & 205 & $3$ \% & $3.33$ \% & $2.86$ \% \\
        Saquinavir		(\textbf{SQV}) & 824  & 206 & $3$ \% & $3.45$ \% & $0.00$ \% \\
%\bottomrule
\end{tabular}
\label{table: Overview of the HIV-1 dataset (PI)}
\end{table}
%%%%%%%%%%%%%%%%%%%%%%
%%%%%%%%%%%%%%%%%%%%%%
%%%%%%%%%%%%%%%%%%%%%%
%
%%%%%%%%%%%%%%%%%%%%%%%%%%%%
%%%%%%%%%%%%%%%%%%%%%%%%%%%%
%%%%%%%%%%%%%%%%%%%%%%%%%%%%
\begin{figure}[t]
  \centering
  \hspace*{-0.6em}
  		\scalebox{1}{
  			\includegraphics[width=1\linewidth]{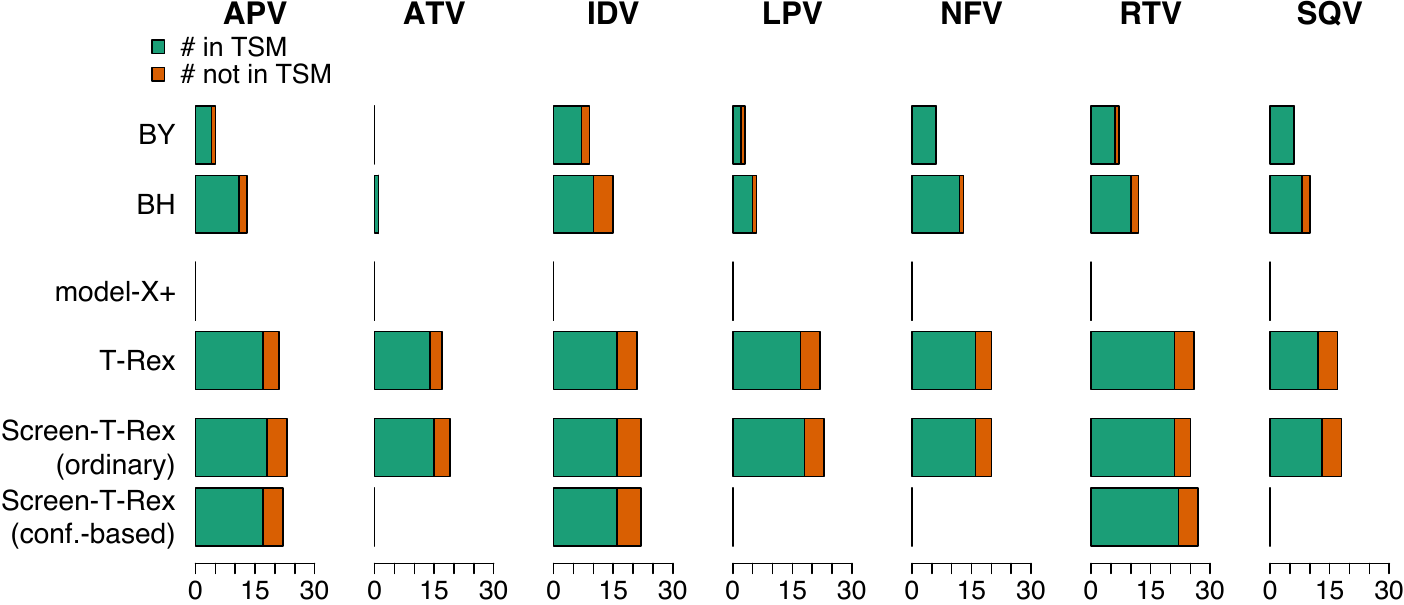}
  		}
   		\label{fig: HIV_PI_drugs}
   		\setlength{\belowcaptionskip}{-8pt}
 \caption{Number of selected mutations that are reported (green) and not reported (orange) in TSM lists for HIV-1 PI-type drugs. The \textit{T-Rex} methods dominate the benchmark methods in terms of the number of selected mutations reported in TSM lists. Moreover, a few potentially relevant mutations that are not reported in TSM lists are detected.}
  \label{fig: HIV_PI_drugs}
\end{figure}
%%%%%%%%%%%%%%%%%%%%%%%%%%%%
%%%%%%%%%%%%%%%%%%%%%%%%%%%%
%%%%%%%%%%%%%%%%%%%%%%%%%%%%
%

\section{Conclusion}
\label{sec: Conclusion}
The \textit{Screen-T-Rex} selector, a fast FDR-controlling variable selection method for large-scale genomics biobanks, was proposed. Using the proposed method in combination with the original \textit{T-Rex} selector, an efficient algorithm for conducting thousands of large-scale GWAS was proposed. In our future research, we will apply the proposed methods to genomics data from the UK biobank.

\cleardoublepage

%Bibliography
\bibliographystyle{IEEEtran}
\bibliography{bibliography}

\typeout{get arXiv to do 4 passes: Label(s) may have changed. Rerun}

\end{document}